\documentclass{ifacconf}

\usepackage{amsmath}
\usepackage{amssymb}
\usepackage{mathrsfs}
\usepackage{tikz,siunitx}
\usetikzlibrary{matrix,positioning,calc,bending}
\usetikzlibrary{decorations.pathmorphing}
\tikzset{snake it/.style={-stealth,
    decoration={snake, 
    amplitude = .4mm,
    segment length = 2mm,
    post length=0.9mm},decorate}}

\tikzstyle{spring}=[thick,decorate,decoration={zigzag,pre length=0.5cm,post
	length=0.5cm,segment length=5}]

\tikzstyle{bloco} = [draw, fill=white, rectangle, minimum height=1cm, minimum width=.5cm, text width=2cm, align = center]
\tikzstyle{rede} = [draw, fill=blue, rectangle, minimum height=1cm, minimum width=.5cm, text width=2cm, align = center]
\tikzstyle{sum} = [draw, fill=white, circle]
\tikzstyle{input} = [coordinate]
\tikzstyle{output} = [coordinate]
\tikzstyle{pinstyle} = [pin edge={to-,thin,black}]
\usetikzlibrary{fit, shapes, arrows}
\usetikzlibrary{arrows.meta}
\usetikzlibrary{calc,patterns,angles,quotes}
\usetikzlibrary{decorations.pathmorphing}
\usetikzlibrary {positioning}

\tikzstyle{startstop} = [rectangle, rounded corners, minimum width=3cm, minimum height=1cm,text centered, draw=black, fill=red!30]
\tikzstyle{io} = [trapezium, trapezium left angle=70, trapezium right angle=110, minimum width=3cm, minimum height=1cm, text centered, draw=black, fill=blue!30]
\tikzstyle{process} = [rectangle, minimum width=3cm, minimum height=1cm, text centered, draw=black, fill=orange!30]
\tikzstyle{decision} = [diamond, minimum width=3cm, minimum height=1cm, text centered, draw=black, fill=green!30]
\tikzstyle{arrow} = [thick,->,>=stealth]

\tikzset{
	block/.style = {draw, fill=blue!20, rectangle, minimum height=3em, minimum width=5em},
	block1/.style = {draw, fill=blue!20, rectangle, minimum height=5em, minimum width=5em},
	bloco/.style = {draw, rectangle, minimum height=1cm, minimum width=1cm},
	network/.style = {draw, dashed, fill=orange!20, rectangle, minimum height=6em, minimum width=3em},
	input/.style = {coordinate,node contents={}},
	output/.style = {coordinate,node contents={}},
	switch/.style = {minimum size=3em,
		path picture={  \draw
			([xshift=-3mm]  \ppbb.center)  -- ++ (45:6mm);
			\draw[shorten >=3mm,-]
			(\ppbb.west) edge (\ppbb.center)
			(\ppbb.east)  --  (\ppbb.center);
			\draw[<->]
			([yshift=-2mm] \ppbb.center) arc[start angle=-15, end angle=75, radius=6mm];   
		},
		label={[yshift=-2mm] above:#1},
		node contents={}},
}

\usepackage[inline]{enumitem}

\usepackage{graphicx}      
\usepackage{natbib}        

\newtheorem{theorem}{Theorem}
\newtheorem{lemma}{Lemma}  
\newtheorem{remark}{Remark} 
\newtheorem{corollary}{Corollary}

\begin{document}
\begin{frontmatter}

\title{Dynamic Event-Triggered Control of Discrete-Time Nonlinear Systems based on Difference-Algebraic Representations\thanksref{footnoteinfo}} 

\thanks[footnoteinfo]{This work was supported by the Brazilian agencies CNPq (Grant numbers: 407885/2023-4, 307758/2025-7; 308791/2025-8), CAPES, and FAPEAM.}

\author[PPGEE]{Vitoriano Medeiros Casas},
\author[IFSEMG]{Gabriela L\'{\i}gia Reis},
\author[UERJ]{Pedro Henrique Silva Coutinho},
\author[DEUFAM]{Iury Bessa},
\author[UEA]{Rodrigo Farias Ara\'ujo}

\address[PPGEE]{Graduate Program in Electrical Engineering, Federal University of Amazonas, Av. Gen. Rodrigo Oct\'avio, 6200, Manaus, AM, Brazil (e-mail: vitoriano.casas@ufam.edu.br)}
\address[IFSEMG]{Department of Technological Education, Federal Institute of Southeast Minas Gerais, R. Bernardo Mascarenhas, 1283, Juiz de Fora, MG, 36080-001, Brazil (e-mail: gabriela.reis@ifsudestemg.edu.br)}
\address[UERJ]{Department of Electronics and Telecommunication Engineering, State University of Rio de Janeiro, Av. S\~ao Francisco Xavier, 524, Rio de Janeiro, RJ, 20550-900, Brazil (e-mail: phcoutinho@eng.uerj.br)}
\address[DEUFAM]{Faculty of Electrical and Computer Engineering, Federal University of Amazonas, Av. Gen. Rodrigo Oct\'avio, 6200, Manaus, AM, Brazil (e-mail: iurybessa@ufam.edu.br)}
\address[UEA]{Department of Control and Automation Engineering, Amazonas State University, Av. Darcy Vargas, 1200, Parque 10 de Novembro, Manaus, AM, 69050-020, Brazil (e-mail: rfaraujo@uea.edu.br)}

\begin{abstract}                
This paper addresses the dynamic event-triggered control for a class of discrete-time nonlinear systems described by a difference-algebraic representation (DAR), using a gain-scheduled controller.  An outstanding aspect of the proposed method is the incorporation of information about the system's nonlinearities into the control law and the trigger function. The proposed event-triggered mechanism also incorporates information on the asynchronous terms induced by the event-based sampling. All these ingredients enable the derivation of a less conservative co-design condition for the simultaneous design of the gain-scheduled control law and the dynamic triggering mechanism to ensure the asymptotic stability of the closed-loop system. An estimate of the region of attraction of the origin of the closed-loop system is obtained to guarantee the closed-loop system's operation within the domain of validity of the DAR.  Then, an optimization problem is formulated to reduce the number of events and enlarge the estimated region of attraction. Finally, the effectiveness of the proposed condition is illustrated by a numerical example.
\end{abstract}


\end{frontmatter}

\section{Introduction}
Event-triggered control (ETC) is a relevant strategy for providing efficient implementations of networked control systems (NCS). Unlike traditional time-triggered approaches, which perform periodic or aperiodic transmissions regardless of the system's state, ETC performs control updates only when certain conditions related to the dynamic behavior of the system are violated~\cite{Abdelrahim2018}.

The main motivation behind ETC lies in reducing the rate of transmitted data over the communication network, relieving congestion, and allowing the control scheme to operate effectively even in the presence of delays and packet losses in shared environments. The key component of ETC systems is the Event-Triggering Mechanism (ETM), which is responsible for monitoring the system and determines when a new data transmission is required. This mechanism evaluates the plant measurements and determines transmission times based on previously defined triggering conditions \cite{Abdelrahim2018}.

The ETM can be classified as static, adaptive, or dynamic. Dynamic ETC incorporates an internal variable whose growth rate depends on the system's state, helping to suppress unnecessary transmissions more efficiently \cite{Girard2025}. For discrete-time systems, an increasing number of works have addressed event-triggered control \cite{Hu2016,Xiao2021}. Dynamic mechanisms have been widely studied for discrete-time systems, particularly in linear settings, where they are known to extend inter-event intervals while preserving stability \cite{Girard2025,Abdelrahim2018}. \cite{Xu2023} proposes an observer-based dynamic event-triggered control strategy for discrete-time linear multi-agent systems. In \cite{Li2021}, adaptive control problems for unknown second-order nonlinear multi-agent systems are investigated. In \cite{Zuo2019,COUTINHO2021,CoutinhoP2022}, the co-design conditions for dynamic ETC of nonlinear systems are presented based on a gain-scheduling approach to quasi-LPV models. Despite these contributions to discrete-time systems, most studies focus on linear dynamics, and there remains a gap in the study of nonlinear systems.

Nonlinear control strategies can often improve closed-loop system performance. One common approach to handling nonlinear systems involves using polytopic inclusions, where the nonlinear dynamics are represented by a convex combination of linear vertex systems, with weighting functions that depend on parameters or scheduling variables evolving within a predefined convex region.
This representation allows the use of linear matrix inequalities (LMIs) constraints and convex optimization techniques to design controllers that guarantee stability within the polytopic region.

Among the various representations of nonlinear systems by means of polytopic inclusions, the differential-algebraic representation has attracted attention due to its capability to model a wide range of real-world nonlinear systems and phenomena \cite{Oliveira2013}. While standard polytopic models approximate nonlinear dynamics through convex combinations of linear systems, the differential-algebraic representation is capable of providing equivalent representations of rational systems through a set of algebraic and differential equations \cite{Reis2021}. 
In the discrete-time context, difference-algebraic representations (DAR) have been proposed by \cite{Coutinho2009}. These representations allow the use of a polytopic formulation for nonlinear systems and the development of LMI-based conditions derived from the application of Lyapunov theory.

Many works have employed DAR for controlling nonlinear systems. DAR-based formulations have been applied to estimate the region of attraction of nonlinear systems with saturating actuators \cite{Coutinho2010} and to address robust stability under uncertainties \cite{pereira_coutinho_2009}. In the context of discrete-time systems, \cite{Oliveira01052013} has proposed a condition for the stability of nonlinear systems subject to disturbances and actuator saturation. DAR has also been applied to the region of attraction (RoA) estimation
\cite{Oliveira01052013,Oliveira2012,Azizi02012018}. Few works have addressed the ETC of nonlinear systems using DAR. In \cite{Moreira201715307}, a method for static ETC design of continuous-time rational systems is proposed, employing an emulation-based approach. In a different approach, \cite{Moreira20202720} develops a co-design condition. More recently, the dynamic ETC has been studied in~\cite{reis2025dynamic}. However, it still lacks developments for ETC design based on DAR for discrete-time nonlinear systems, particularly concerning dynamic ETC, which is the problem pursued in this paper.

A central issue in applying ETC to nonlinear controllers, especially gain-scheduling schemes, is asynchronism, which induces a mismatch between the actual current state of the system and the last state received by the controller. To incorporate the effects of asynchronism in the polytopic representation, a cancellation-based approach is adopted in \cite{CoutinhoP2022,CoutinhoP2022110292,reis2025dynamic,PessimP2023,coutinho2025resilient}, where an additional term is introduced in the triggering function to mitigate asynchronism. In this context, this paper proposes the following: 
\begin{enumerate*}[label=(\Roman*)]
    \item A dynamic ETC for discrete-time nonlinear rational systems in DAR form that includes an additional term to deal with the asynchronism;
    \item Co-design conditions for a dynamic ETM and a nonlinear controller. The controller is designed to utilize nonlinearity information and is capable of asymptotically stabilizing the systems under polytopic inclusions;
    \item An optimization problem focused on minimizing the number of events while simultaneously expanding the region of attraction.
\end{enumerate*}

The remainder of this paper is structured as follows. The problem formulation is stated in Section~\ref{sec:problem_formulation}. Section~\ref{sec:main_results} develops the co-design methodology for the proposed ETC strategy at the same time that the estimate of the region of attraction is maximized. Section~\ref{sec:numerial_results} presents numerical experiments that illustrate the effectiveness of the proposed method and validate the co-design conditions. Concluding remarks are provided in Section~\ref{sec:conclusion}.

\textbf{Notation}. $\mathbb{N}$ is the set of positive integers, $\mathbb{N}_{\leq p}$ the set of positive integers less than or equal to $p \in \mathbb{N}, \mathbb{N}_0$ is the set of nonnegative integers, $\mathbb{B}=\{0,1\}$ is the Boolean domain, $\mathbb{R}^n$ is the $n$-dimensional Euclidean space, $\mathbb{R}^{n \times m}$ is the set of $n \times m$ real matrices, and $\mathbb{R}_{\geq 0}\left(\mathbb{R}_{>0}\right)$ is the set of nonnegative (positive) real numbers. For a symmetric block matrix, the symbol $\ast$ denotes the term deduced by symmetry, $\operatorname{diag}(\cdot)$ stands for a block-diagonal matrix, and $\operatorname{He}(M)$ represents $M+M^{\top}$. $\lambda_{\min }(M)\left(\lambda_{\max }(M)\right)$ denotes the minimum (maximum) eigenvalue of matrix $M$. $I$ and 0 denote the identity and null matrices of appropriate dimensions, respectively. Given a multi-index $\mathbf{i}=\left(i_1, \ldots, i_p\right) \in \mathbb{B}^p$, where $\mathbb{B}^p=\left\{\mathbf{i}: i_j \in \mathbb{B}, j \in \mathbb{N}_{\leq p}\right\}$, it is defined $\mathbb{B}^{p+}=\left\{\mathbf{i}: i_j \leq i_{j+1}, i_j \in \mathbb{B}, j \in \mathbb{N}_{p-1}\right\}$ and $\mathscr{P}(\mathbf{i})$ is the set of permutations of the entries of $\mathbf{i}$.

\section{Problem formulation} \label{sec:problem_formulation}

Next, we see a few subsections. Consider the diagram of the NCS with ETM shown in Fig.~\ref {fig:diagram_network}. In the considered setup, the plant is a nonlinear system whose state $x_k$, indexed by $ k \in \mathbb{N}_0$, is transmitted to the controller via a general-purpose communication network with transmission instants determined by the triggering policy of the ETM. The ETM generates a sequence of transmission instants $\{k_j\}_{j=0}^{\infty}$, where $k_j$ is the $j$-th transmission instant.
Then, the state information available to the controller, $\hat{x}_k$, is held constant, through a zero-order hold (ZOH) mechanism, until the next transmission instant.

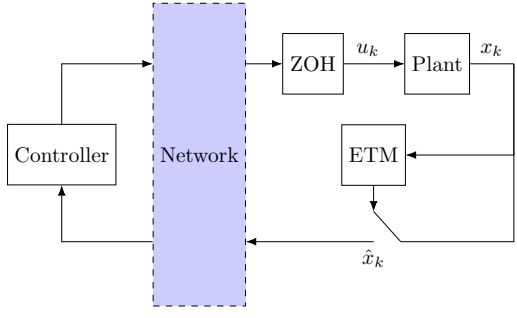
\begin{figure}[!ht]
    \centering
    \scalebox{.8}{
    \begin{tikzpicture}[auto, node distance=2cm,>={Latex[flex]}]
    \node (model_inp)   [input] {};
    \node (zoh)         [bloco, right=.6cm of model_inp] {ZOH};
    \node (model)       [bloco, right=1cm of zoh] {{Plant}};
    \node (model_out)   [right=.6cm of model] {};
    \node (etm)         [bloco, below=1cm of zoh, xshift=1cm] at (zoh) {ETM};
    
    \node (switch_out)  [below=.8cm of etm] {};
    \node (switch_in)   [right=0.2cm of switch_out] {};
    \node (etm_out_switch) [below=0.3cm of etm] {};
    \node (network)     [bloco, dashed, fill=blue!20, left=0cm of model_inp, yshift=-1.5cm, minimum height=5cm] {Network};
    \node (controller)  [bloco, left=.6cm of network] {Controller};

    \node (switch_out2) [left=.1cm of switch_out] {};
    \node (etm_out_switch2) [below=0cm of etm] {};
    \node (etm_out_switch3) [left=0.1cm of etm_out_switch2] {};

    \node (network_left) [left=0cm of network] {};

    \path let \p1 = (model_inp), \p2 = (switch_out) in
        node (network_in) [] at (\x1,\y2) {};
    \path let \p1 = (network_left), \p2 = (switch_out) in
        node (network_out) [] at (\x1,\y2) {};
    \path let \p1 = (network_left), \p2 = (model_inp) in
        node (network_model_out) [] at (\x1, \y2) {};

    \draw [->] (model_inp.east) -- node[pos=.3] (y) {} (zoh.west);
    \draw [->] (zoh.east) -- node[pos=.4] (y) {$u_k$} (model.west);
    \draw [-] (model.east) -- node[pos=0.5] (y) {$x_k$} (model_out.center);
    \draw [->] (model_out.center) |- node[pos=0.5] (y) {} (etm.east);
    \draw [-] (model_out.center) |- node[pos=0] (y) {} (switch_in.center);
    \draw [-] (switch_in.center) -- node[pos=0] {} (etm_out_switch.center);
    \draw [->] (etm.south) -- node[pos=0] {} (etm_out_switch.center);
    
    \draw [->] (switch_out.center) -- node[pos=0] {$\hat{x}_k$} (network_in.center);
    \draw [->] (network_out.east) -| node[pos=0] {} (controller.south);
    \draw [->] (controller.north) |- node[pos=.8] {} (network_model_out.east);
    \end{tikzpicture}}
    \caption{Networked control system with ETC scheme.}
    \label{fig:diagram_network}
\end{figure}

Consider that the plant is described by a discrete-time nonlinear dynamical system as follows:
\begin{equation} \label{eq:affine_system}
    x_{k+1} = f(x_k) + g(x_k) u_k
\end{equation}
where $x_k \in \mathcal{D}_x \subseteq \mathbb{R}^n$ is the state vector of the system and $u_k \in \mathbb{R}^m$ is the control input. The functions $f : \mathbb{R}^n \to \mathbb{R}^n$, with $f (0) = 0$, and $g : \mathbb{R}^n \to \mathbb{R}^{n \times m}$, are assumed rational and well-defined on $\mathcal{D}_x$. Then,  the system~\eqref{eq:affine_system} can be reformulated as a DAR described by~\cite{Reis2021,Coutinho2010}:
\begin{equation} \label{eq:system_dar}
    \begin{aligned}
     x_{k+1} &= A_1(x_k) x_k + A_2(x_k) \pi_k + A_3(x_k) u_k,  \\
     0 &= \Omega_1(x_k) x_k + \Omega_2(x_k) \pi_k + \Omega_3(x_k) u_k, \\
    \end{aligned}
\end{equation}
where $\pi_k \in \mathbb{R}^p$ is the vector of nonlinearities. The matrix-valued functions $A_1(x) \in \mathbb{R}^{n \times n}, A_2(x) \in \mathbb{R}^{n \times p}$, $A_3(x) \in \mathbb{R}^{n \times m}, \Omega_1(x_k) \in \mathbb{R}^{p \times n}, \Omega_2(x) \in \mathbb{R}^{p \times p}$ and $\Omega_3(x) \in \mathbb{R}^{p \times m}$ are affine functions in $x$ and the matrix $\Omega_2(x)$ is square and invertible, for all $x \in \mathcal{D}_x$. The domain of validity of the DAR is given by the following compact polyhedral set:
\begin{equation}\label{eq:domain_of_validity}
    \mathcal{D}_x = \biggl\{ x_k \in \mathbb{R}^n : b_j^\top x_k \leq 1, j \in \mathbb{N}_{\leq h} \biggl\}
\end{equation}
where $b_j \in \mathbb{R}^n$, $j \in \mathbb{N}_{\leq h}$, define the hyperplanes.

Consider the vector $\pi_k$ written as the solution of the algebraic equation in~\eqref{eq:system_dar}:
\begin{equation} \label{eq:pi_omega_2}
    \pi_k = -\Omega_2^{-1}(x_k)\left[\Omega_1(x_k) x_k+\Omega_3(x_k) u_k\right].
\end{equation}
The equivalence between the nonlinear system~\eqref{eq:affine_system} and its DAR~\eqref{eq:system_dar} can be verified by substituting~\eqref{eq:pi_omega_2} in the dynamical equation of~\eqref{eq:system_dar}, which results in~\eqref{eq:affine_system}.

Provided that the matrix-valued functions of~\eqref{eq:system_dar} are affine in $x$ and $\mathcal{D}_x$ is a compact set,
each state-dependent element can be bounded by $z_j^0 \leq z_j (x) \leq z_j^1$, for all $j \in \mathbb{N}_{\leq r}$, where $r$ denotes the number of scheduling variables, $z_{j}(x) : \mathcal{D}_x \to \mathbb{R}$ denotes the $j$-th scheduling variable, and $z_j^0, z_j^1$ are its lower and upper bounds, respectively. Then, the DAR system~\eqref{eq:system_dar} can be equivalently represented as the following polytopic model: 
\begin{equation}\label{eq:polytopic_model}
\begin{aligned}
    x_{k+1} &= \sum_{\mathbf{i} \in \mathbb{B}^r} \alpha_{\mathbf{i}}(x_k) (A_{1\mathbf{i}}x_k + A_{2\mathbf{i}}\pi_k + A_{3\mathbf{i}} u_k) \\
    0 &= \sum_{\mathbf{i} \in \mathbb{B}^r} \alpha_{\mathbf{i}}(x_k) (\Omega_{1\mathbf{i}}x_k + \Omega_{2\mathbf{i}}\pi_k + \Omega_{3\mathbf{i}} u_k),
\end{aligned}
\end{equation}
where $\alpha_{\mathbf{i}}(x)$ is given by
\begin{equation}
    \alpha_{\mathbf{i}}(x_k) = \prod_{j=1}^r w_{i_j}^j (x_k), 
\end{equation}
where $\mathbf{i} = (i_1, ...,i_r) \in \mathbb{B}^r$ is an $r$-dimensional multi-index, and the weight functions are given by
\begin{equation}
    w_{j}^0(x_k) = \frac{z_j ^1 - z_j (x_k)}{z_j ^1 - z_j ^0},~ w_j^1(x_k) = 1 - w_j^0(x_k).
\end{equation}
Thus, the parameters satisfy the convex sum property
\begin{equation}
    \sum_{\mathbf{i}\in \mathbb{B}^r} \alpha_{\mathbf{i}}(x) = 1, \, \text{and} \, \alpha_{\mathbf{i}}(x) \geq 0, \quad \forall \mathbf{i} \in \mathbb{B}^r, \, \forall x \in \mathcal{D}_x.
\end{equation}

In this work, we employ the following event-triggered gain-scheduled nonlinear control law:
\begin{equation} \label{eq:control_law}
    u_k = K(\hat{x}_k) \hat{x}_k + L(\hat{x}_k) \hat{\pi}_k,
\end{equation}
where the state-dependent gains are parameterized as
\begin{equation} \label{eq:control_law_scheduled}
    \begin{aligned} 
    K(\hat{x}) &= \sum_{\mathbf{j} \in \mathbb{B}^r} \alpha_{\mathbf{j}}(\hat{x}) K_{\mathbf{j}}, \text{ and }
    L(\hat{x}) = \sum_{\mathbf{j} \in \mathbb{B}^r} \alpha_{\mathbf{j}}(\hat{x}) L_{\mathbf{j}},
    \end{aligned}
\end{equation}
with $L_{\mathbf{j}}\in \mathbb{R}^{m\times p}, K_{\mathbf{j}} \in \mathbb{R}^{m\times n}$, for all $\mathbf{j} \in \mathbb{B}^r$, being the gains to be designed. Moreover, due to the ZOH mechanism, the information on the state $\hat{x}_k$ and the vector of nonlinearities $\hat{\pi}_k$ available to the controller are
\begin{equation}
    \hat{x}_k=x_{k_j}, \quad \hat{\pi}_k = \pi_{k_j}, \quad \forall k \in \mathcal{I}_j,
\end{equation}
where $\mathcal{I}_j := \left\{k_j, \ldots, k_{j+1}-1\right\}$ denotes the discrete-time indices between two consecutive transmission instants $k_j$ and $k_{j+1}$. Then, the transmission errors induced by the event-based sampling are
\begin{equation} \label{eq:transmission_error}
    \begin{cases} e_k= \hat{x}_k-x_k,~\forall k \in \mathcal{I}_j, \\
    \delta_k= \hat{\pi}_k-\pi_k, ~\forall k \in \mathcal{I}_j. \end{cases}
\end{equation}
Consider the augmented vector $\zeta = \begin{bmatrix}
x^\top & \pi^\top & e^\top & \delta^\top \end{bmatrix}^\top$.
By substituting the control law~\eqref{eq:control_law} into the DAR~\eqref{eq:system_dar}, 
and using the transmission errors~\eqref{eq:transmission_error}, the following closed-loop system 
can be obtained\footnote{Hereafter, the discrete-time index is omitted for brevity unless 
explicitly required.}:
\begin{equation}\label{eq:dar_nl_law}
    \begin{aligned}
    x_{k+1} = x_+ & =\varphi_1 \left(\zeta \right)+\xi_1 \left(\zeta \right), \\
    0 & =\varphi_2\left(\zeta \right)+ \xi_2 \left(\zeta\right),
    \end{aligned}
\end{equation}
where
\begin{align*}
\varphi_1(\zeta) 
&= \left(A_1(x) + A_3(x) K(x)\right) x
   + A_3(x) K(x) e \notag \\
&\quad + \left(A_2(x) + A_3(x)L(x)\right) \pi
   + A_3(x) L(x) \delta, \\[0.5em]
\varphi_2(\zeta) 
&= \left(\Omega_1(x) + \Omega_3(x) K(x)\right) x
   + \Omega_3(x) K(x) e \notag \\
&\quad + \left(\Omega_2(x) + \Omega_3(x)L(x)\right) \pi
   + \Omega_3(x) L(x) \delta, \\[0.5em]
\xi_1(\zeta) 
&= A_3(x)\!\left\{(K(x{+}e)-K(x))(x{+}e) \right. \notag \\
&\quad \left. + (L(x{+}e)-L(x))(\delta + \pi) \right\}, \\[0.5em]
\xi_2(\zeta) 
&= \Omega_3(x)\!\left\{(K(x{+}e)-K(x))(x{+}e) \right. \notag \\
&\quad \left. + (L(x{+}e)-L(x))(\delta + \pi) \right\}.
\end{align*}

\begin{remark}
Different from the polytopic model of the system~\eqref{eq:polytopic_model}, whose state-dependent parameters depend on the plant state $x_k$, the state-dependent gains of the gain-scheduled nonlinear controller~\eqref{eq:control_law} depend on the sampled state $\hat{x}_k$. This leads to the so-called asynchronous phenomenon \cite{coutinho2025resilient,COUTINHO2021,CoutinhoP2022,reis2025dynamic}.  The terms induced by the asynchronous phenomenon in the closed-loop dynamics are $\xi_1(\zeta)$ and $\xi_2(\zeta)$. Thus, it is clear that this phenomenon induces an additional internal perturbation to the closed-loop system.
If this phenomenon is not properly handled, the gain-scheduling structure of the controller~\eqref{eq:control_law} degenerates to a linear control structure, that is, $L_{\mathbf{j}} \approx L$, $K_{\mathbf{j}} \approx K$, $\forall \mathbf{j} \in \mathbb{B}^{r}$. The solution to deal with this issue will be discussed in the next section.
\end{remark}

Based on the previous discussion, this work aims to address the co-design problem of determining a gain-scheduled nonlinear control law~\eqref{eq:control_law_scheduled} operating under a dynamic ETM for ensuring that the origin of the closed-loop system~\eqref{eq:dar_nl_law} is asymptotically stable.
Moreover, the proposed approach aims to reduce the number of transmissions generated by the ETM and 
estimate the region of attraction of the origin of~\eqref{eq:dar_nl_law} within the domain of validity $\mathcal{D}_x$ of the DAR model given in~\eqref{eq:domain_of_validity}.

\section{Main results} \label{sec:main_results}

The main contributions are presented in this section. First, the dynamic ETC scheme and an appropriate trigger function are proposed to compensate for the asynchronous phenomenon between the plant and the controller. Then, a co-design condition is presented to obtain the controller and the ETM parameters. Moreover, the region of attraction of the origin of the closed-loop system is estimated. Finally, an optimization problem is formulated to reduce the number of transmissions and enlarge the estimated region of attraction.

\subsection{Dynamic ETM}

To reduce the number of transmissions and the usage of communication resources, the transmission instants are determined by the following dynamic ETM:
\begin{equation}\label{eq:dynamic_etm_pi}
    k_{j+1}  = \min\{ k \in \mathbb{N} : k > k_j  \wedge \eta_k + \theta \Gamma(\zeta_k) < 0 \},
\end{equation}
for all $j \in \mathbb{N}$, where $k_0 = 0$, $\theta \in \mathbb{R}_{>0}$ is a given parameter, 
and $\eta_k \in \mathbb{R}_{\geq 0}$ is the internal dynamics of the ETM, which evolves according to
\begin{equation}\label{eq:dynamic_variable_eta_pi}
    \eta_{k+1} = (1-\lambda)\eta_k + \Gamma(\zeta_k), \quad \forall k \in \mathcal{I}_j,
\end{equation}
where $\lambda \in \mathbb{R}_{>0}$ is a parameter related to the decay rate of $\eta_k$. Note that, by construction, the Zeno phenomenon is avoided, since the minimum interval between events is given by a sampling period.

The dynamic ETM~\eqref{eq:dynamic_etm_pi}--\eqref{eq:dynamic_variable_eta_pi} has a trigger function described by
\begin{multline} \label{eq:trigger_function}
    \Gamma(\zeta)=x^{\top} Q_x x + \pi^{\top} Q_{\pi} \pi
    - e^{\top} Q_e e - \delta^\top Q_\delta \delta -\xi(\zeta),
\end{multline}
where
\begin{equation} \label{eq:async_term}
    \xi (\zeta) = 2 \varphi_1^T P \xi_1 + \xi_1^T P \xi_1 + 2\pi^T Z \xi_2,
\end{equation}
The term $\xi(\zeta)$ is to be designed to handle the asynchronous phenomenon and to allow the development of a suitable co-design condition. The term $x^{\top} Q_x x - e^{\top} Q_e e$ can be viewed as the deviation between the current state $x_k$ and the last transmitted state $\hat{x}_k$, while $\pi^{\top} Q_{\pi} \pi - \delta^\top Q_\delta \delta$ can be viewed as the deviation between the current vector of nonlinearities $\pi_k$ and the last transmitted vector of nonlinearities $\hat{\pi}_k$. Thus, the dynamic ETC scheme determines the appropriate transmission instants of the states in Fig.~\ref{fig:diagram_network}.

The next lemma establishes the positiveness of $\eta$ according to~\eqref{eq:dynamic_variable_eta_pi}. This property is required to construct a suitable Lyapunov function candidate to derive the co-design condition.
\begin{lemma} \label{lem:eta_lemma}
    Consider the dynamic ETC scheme~\eqref{eq:dynamic_etm_pi}--\eqref{eq:trigger_function} with initial condition $\eta_0 \geq 0$. If $\theta \geq 1/(1-\lambda)$, then $\eta_k~\geq~0$, $\forall k \in \mathbb{N}_0$.
\end{lemma}
\begin{pf}
    The proof follows from \cite[Lemma~1]{Hu2016}. By construction, it is not difficult to see that~\eqref{eq:dynamic_etm_pi} guarantees that $\eta_k + \theta \Gamma (\zeta) \geq 0,  \forall k \in \mathbb{N}$. Note that when $\theta = 0$, $\eta_k \geq 0$. If $\theta \neq 0$, and using~\eqref{eq:dynamic_variable_eta_pi}, then $\eta_{k+1} \geq \left( 1 - \frac{1}{\theta} - \lambda \right) \eta_k$. Using the last inequality iteratively and noting that $\eta_0 \geq 0$, it follows that $\eta_k \geq \left( 1 - \frac{1}{\theta} - \lambda \right)^{k} \eta_0 \geq 0$, $\forall k \in \mathbb{N}$. Provided that $\eta_0 \geq 0$, then $\eta_k \geq 0$ $\forall k \in \mathbb{N}_0$, under $\theta \geq 1/(1-\lambda)$. This completes the proof.~\hfill~$\blacksquare$
\end{pf}

\subsection{Co-design condition}

The condition to co-design the control gains and the parameters of the dynamic ETM are presented in Theorem~\ref{th:theorem_main}.
\begin{theorem} \label{th:theorem_main}
    Consider the discrete-time dynamical system in~\eqref{eq:affine_system} represented in a DAR form as in~\eqref{eq:system_dar}, the control law~\eqref{eq:control_law}, and the dynamic ETM described in~\eqref{eq:dynamic_etm_pi}--\eqref{eq:trigger_function}. Let $\eta_0, \theta \in \mathbb{R}_{\geq 0}, \lambda \in \mathbb{R}_{> 0}$, being given scalars satisfying $\theta > 1/(1-\lambda)$. If there exist matrices $\tilde L_{\mathbf{j}}\in \mathbb{R}^{m\times p}, \tilde K_{\mathbf{j}} \in \mathbb{R}^{m \times n}, \mathbf{j} \in \mathbb{B}^r, \tilde Z \in \mathbb{R}^{p \times p}$, and symmetric positive definite matrices $\tilde Q_x, \tilde Q_e,X \in \mathbb{R}^{n \times n}$ and $\tilde Q_{\pi}, \tilde Q_{\delta} \in \mathbb{R}^{p \times p}$, such that the following LMIs hold:
    \begin{equation} \label{eq:attraction_region_lmi_theorem1}
        \begin{bmatrix}
            1 & b_j^\top X \\
            \ast & X
        \end{bmatrix} \geq 0, \quad j \in \mathbb{N}_{\leq h}, 
    \end{equation}
    \begin{equation} \label{eq:lmi_sum_theorem1}
        \sum_{(\mathbf{i}, \mathbf{j})\in \mathscr{P}(\mathbf{m},\mathbf{n})} \Phi_{\mathbf{i} \mathbf{j}} < 0, \quad \forall \mathbf{m},\mathbf{n} \in \mathbb{B}^{r+},  
    \end{equation}
    with
    \begin{align*} 
        \Phi_{\mathbf{i} \mathbf{j}} &= 
        \begin{bmatrix}
        -X & 0 & \Phi_{1,3} & 0 & \Phi_{1,5} & X & 0 \\
        \ast & -\tilde Q_e & \tilde K_{\mathbf{j}}^\top \Omega_{3 \mathbf{i}}^\top & 0 & \tilde K_{\mathbf{j}}^\top A_{3 \mathbf{i}}^\top & 0 & 0\\
        \ast & \ast & \Phi_{3,3} & \Omega_{3 \mathbf{i}} \tilde L_{\mathbf{j}}^\top & \Phi_{3,5} & 0 & \tilde Z^\top \\
        \ast & \ast & \ast & - \tilde Q_{\delta} & \tilde L_{\mathbf{j}}^\top A_{3 \mathbf{i}}^\top & 0 & 0 \\
        \ast & \ast & \ast & \ast & -X & 0 & 0 \\
        \ast & \ast & \ast & \ast & \ast & -\tilde Q_x & 0 \\
        \ast & \ast & \ast & \ast & \ast & \ast & -\tilde Q_{\pi}
        \end{bmatrix}, \\
        \Phi_{1,3} &= X \Omega_{1 \mathbf{i}}^\top + \tilde K_{\mathbf{j}}^\top \Omega_{3 \mathbf{i}}^\top, \quad \Phi_{1,5} = X A_{1 \mathbf{i}}^\top + \tilde K_{\mathbf{j}}^\top A_{3 \mathbf{i}}^\top, \\
        \Phi_{3,3} &= \operatorname{He}(\Omega_{2 \mathbf{i}} \tilde Z + \Omega_{3 \mathbf{i}} \tilde L_{\mathbf{j}}), \quad
        \Phi_{3,5} = \tilde Z^\top A_{2 \mathbf{i}}^\top + \tilde L_{\mathbf{j}}^\top A_{3 \mathbf{i}}^\top,
    \end{align*}
    then, the origin of the closed-loop system~\eqref{eq:dar_nl_law} is asymptotically stable with the control gains given by $K_{\mathbf{j}} = \tilde K_{\mathbf{j}} X^{-1}, L_{\mathbf{j}}=\tilde L_{\mathbf{j}} \tilde Z^{-1}, \mathbf{j} \in \mathbb{B}^r$, and ETM parameters given by $Q_e = X^{-1}\tilde Q_e X^{-1}$, $Q_x= \tilde Q_x^{-1}$, $Q_{\delta} = \tilde Z^{-\top} \tilde Q_{\delta} \tilde Z^{-1}$, $Q_{\pi} = \tilde Q_{\pi}^{-1}$, $P = X^{-1}$, $Z = \tilde Z^{-\top}$. Moreover, the Lyapunov function that certifies the asymptotic stability of the origin is
    \begin{equation} \label{eq:lyapunov_candidate}
        W(x,\eta) = V(x) + \eta, \quad V(x) = x^\top P x,
    \end{equation}
    and, for a given $\eta_0$, the state trajectories $(x_k,\eta_k)$ with initial condition $x_0$ taken inside of
    \begin{equation} \label{eq:initial_states_region}
        \mathcal{R}_0 = \{x \in \mathbb{R}^{n} : V(x) \leq 1 -\eta_0, \; \eta_0 < 1 \},
    \end{equation}
    converge asymptotically to the origin $(x,\eta) = (0,0)$ without leaving the region
    \begin{equation}~\label{eq:level_set}
        \mathcal{R} = \{ x \in \mathbb{R}^n, \eta \in \mathbb{R}_{> 0} : W(x, \eta) \leq 1\},
    \end{equation}
    which, in its turn, satisfies $\mathcal{R} \subset \mathcal{D}_x \times \mathbb{R}_{\geq 0}$.
\end{theorem}
\begin{pf}
        Considering the Lyapunov function candidate~\eqref {eq:lyapunov_candidate}, Lemma~\ref{lem:eta_lemma}, and $X$ is a symmetric positive definite matrix. Suppose that the inequalities in~\eqref{eq:lmi_sum_theorem1} hold. Then, it follows that
        \begin{equation} \label{eq:lmi_sum_theorem1_pf}
        \Phi (x) = \sum_{\mathbf{m} \in \mathbb{B}} \sum_{\mathbf{n} \in \mathbb{B}} \alpha_{\mathbf{m}}(x) \alpha_{\mathbf{n}}(x) \left( \sum_{(\mathbf{i}, \mathbf{j})\in \mathscr{P}(\mathbf{m},\mathbf{n})} \Phi_{\mathbf{i} \mathbf{j}} \right) < 0.  
        \end{equation}
        Multiplying~\eqref{eq:lmi_sum_theorem1_pf} by $\operatorname{diag}\{ X^{-1}, X^{-1}, \tilde Z^{-1}, $ $\tilde Z^{-1}, X^{-1}, I, I \}$ on the left and its transpose on the right, since $X > 0$ and $\tilde Z$ is invertible, therefore
        \begin{equation}
        \begin{aligned}
        \Phi_{\mathbf{i} \mathbf{j}}^1 =
            \begin{bmatrix}
            -X^{-1} & 0 & \Phi_{1,3}^1 & 0 & \Phi_{1,5}^1 & I & 0 \\
            \ast & \Phi_{2,2}^1 & \Phi_{2,3}^1 & 0 & \Phi_{2,5}^1 & 0 & 0 \\
            \ast & \ast & \Phi_{3,3}^1 & \Phi_{3,4}^1 & \Phi_{3,5}^1 & 0 & I \\
            \ast & \ast & \ast & \Phi_{4,4}^1 & \Phi_{4,5}^1 & 0 & 0 \\
            \ast & \ast & \ast & \ast & -X^{-1} & 0 & 0 \\
            \ast & \ast & \ast & \ast & \ast & - \tilde Q_x & 0 \\
            \ast & \ast & \ast & \ast & \ast & \ast & - \tilde Q_{\pi}
        \end{bmatrix} < 0,
        \end{aligned}
        \end{equation}
        with
        \begin{equation*}
            \begin{aligned}
                \Phi_{1,3}^1 &= (\Omega_{1 \mathbf{i}}^\top + X^{-1} \tilde K_{\mathbf{j}}^\top \Omega_{3 \mathbf{i}}^\top) \tilde Z^{-1}, \\
                \Phi_{1,5}^1 &= (A_{1 \mathbf{i}}^\top + X^{-1} \tilde K_{\mathbf{j}}^\top A_{3 \mathbf{i}}^\top) X^{-1},~
                \Phi_{2,2}^1 = -X^{-1} \tilde Q_e X^{-1}, \\
                \Phi_{2,3}^1 &= X^{-1} \tilde K_{\mathbf{j}}^\top \Omega_{3 \mathbf{i}}^\top \tilde Z^{-1},~
                \Phi_{2,5}^1 = X^{-1} \tilde K_{\mathbf{j}}^\top A_{3 \mathbf{i}}^\top X^{-1}, \\
                \Phi_{3,3}^1 &= \operatorname{He} (\tilde Z^{-\top} (\Omega_{2 \mathbf{i}} + \Omega_{3 \mathbf{i}} \tilde L_{\mathbf{j}} \tilde Z^{-1})), \\
                \Phi_{3,4}^1 &= \tilde Z^{-\top} \Omega_{3 \mathbf{i}} \tilde L_{\mathbf{j}} \tilde Z^{-1},~
                \Phi_{3,5}^1 = (A_{2 \mathbf{i}}^\top + \tilde Z^{-\top} \tilde L_{\mathbf{j}}^\top A_{3 \mathbf{i}}^\top) X^{-1}, \\
                \Phi_{4,4}^1 &= \tilde Z^{-\top} \tilde Q_{\delta} \tilde Z^{-1},~
                \Phi_{4,5}^1 = \tilde Z^{-\top} \tilde L_{\mathbf{j}}^\top A_{3 \mathbf{i}}^\top X^{-1}.
            \end{aligned}
        \end{equation*}
        Applying the Schur complement and the change of variables $P=X^{-1}$, $Z = \tilde Z^{-\top}$, $K = \tilde K_{j} X^{-1}$, $L_j = \tilde L_j \tilde Z^{-1}$, $Q_x = \tilde Q_x^{-1}$, $Q_e = X^{-1} \tilde Q_e X^{-1}$, $Q_{\delta} = \tilde Z^{-\top} \tilde Q_{\delta} \tilde Z^{-1}$, \mbox{$Q_{\pi} = \tilde Q_{\pi}^{-1}$}, $A_{K} = A_{1 \mathbf{i}} + A_{3 \mathbf{i}} K_{\mathbf{j}}$, $\Omega_{K} = \Omega_{1 \mathbf{i}} + \Omega_{3 \mathbf{i}} K_{\mathbf{j}}$, $A_{L} = A_{2\mathbf{i}} + A_{3\mathbf{i}} L_{\mathbf{j}}$ e $\Omega_{L} = \Omega_{2\mathbf{i}} + \Omega_{3\mathbf{i}} L_{\mathbf{j}}$, it results in
        \begin{align}\label{eq:thr1_pf_last_matrix}
            \begin{bmatrix}
                \Psi_{11} & \Psi_{12} \\
                \star & \Psi_{22}
            \end{bmatrix} < 0,
        \end{align}
        where        
        \begin{align*}
            \Psi_{11} &= 
            \begin{bmatrix}
                A_{K}^\top P A_{K} - P + Q_x & A_{K}^{\top} P A_3 K \\
        \ast &  K^{\top} A_3^{\top} P A_3 K - Q_e 
            \end{bmatrix}, \\
            \Psi_{12} &= 
            \begin{bmatrix}
                A_{K}^{\top} P A_{L} + \Omega_{K}^{\top} Z^{\top} & A_K^\top P A_3 L \\
          K^{\top} A_3^{\top} P A_{L} + K^{\top} \Omega_3^{\top} Z^{\top} & K^{\top} A_3^{\top} P A_3 L
            \end{bmatrix}, \\
            \Psi_{22} &= 
            \begin{bmatrix}
                \operatorname{He}(Z \Omega_{L}) + A_{L}^{\top} P A_{L} + Q_{\pi} & A_{L}^\top P A_3 L + Z \Omega_3 L \\
        \ast & L^\top A_3^\top P A_3 L - Q_{\delta}
            \end{bmatrix}.
        \end{align*}
        Multiplying~\eqref{eq:thr1_pf_last_matrix} by $\zeta^\top$ on the left and its transpose on the right, then adding the term $\xi (\zeta)$, given by~\eqref{eq:async_term}, on both sides of the inequality, it follows that
        \begin{multline}
            (\varphi_1 + \xi_1)^T P (\varphi_1 + \xi_1) - x^T P x \\
            + {2\pi^T Z (\varphi_2 + \xi_2)} + x^T Q_x x + \pi^{\top} Q_{\pi} \pi \\
            - e^T Q_e e - \delta^\top Q_\delta \delta - \xi(\zeta) < 0.
        \end{multline}
        Using the DAR in~\eqref{eq:dar_nl_law}, the trigger function~\eqref{eq:trigger_function}, and Lemma~\ref{lem:eta_lemma}, for some given $\eta_0, \theta \in \mathbb{R}_{\geq 0}$, $\lambda \in \mathbb{R}_{> 0}$, leads to
        \begin{equation}
            x_{+}^\top P x_{+} - x^\top P x - \lambda \eta_k + \Gamma(\zeta) < 0,
        \end{equation}
        or, equivalently, 
        \begin{equation}\label{eq:lyap_diff}
            \Delta W(x,\eta) := x_{+}^\top P x_{+} - x^\top P x + {\eta_{k+1} - \eta_k} < 0.
        \end{equation}
        The above inequality ensures $\Delta W(x,\eta) < 0$, $\forall(x,\eta) \neq (0,0)$, and \eqref{eq:lyapunov_candidate} is a Lyapunov function that certifies the asymptotic stability of the origin of the closed-loop system~\eqref{eq:dar_nl_law}. It then follows from~\eqref{eq:lyap_diff} that $W(x_k, \eta_k) \leq W(x_0,\eta_0) = V(x_0) + \eta_0$.
        Then, if $x_0 \in \mathcal{R}_0$, for some $\eta_0 \in [0,1)$, it follows that $W(x_k, \eta_k) \leq 1$, which means that $(x_k,\eta_k) \in \mathcal{R}$, $\forall k \in \mathbb{N}_0$.
        
        By applying a congruence transformation to~\eqref{eq:attraction_region_lmi_theorem1} with the matrix $\operatorname{diag}\{1, X^{-1}\}$, then multiplying the result by $\begin{bmatrix}
            1 & x^\top
        \end{bmatrix}$ on the right and its transpose on the left, it follows that $1 - 2b_j^\top x + V(x) \geq 0$, $\forall j \in \mathbb{N}_{\leq 0}$. By summing and subtracting $\eta$, yields
        $1 - 2b_j^\top x + W(x,\eta) - \eta \geq 0$, $\forall j \in \mathbb{N}_{\leq 0}$. Then, if $(x,\eta) \in \mathcal{R}$, it follows that $2 - 2b_j^\top x - \eta \geq 0$, $\forall j \in \mathbb{N}_{\leq 0}$, which implies from the fact that $\eta_k \geq 0$, $\forall k \in \mathbb{N}_0$, that $b_j^\top x \leq 1$, $\forall j \in \mathbb{N}_{\leq 0}$, and $x \in \mathcal{D}_x$. This implies that $\mathcal{R} \subset \mathcal{D}_x \times \mathbb{R}_{\geq 0}$. 
        
        Thus, if $x_0 \in \mathcal{R}_0$ for some $\eta_0 \in [0,1)$, then $(x_k, \eta_k)$ converge asymptotically to the origin $(x,\eta) = (0,0)$, without leaving the region $\mathcal{R} \subset \mathcal{D}_x \times \mathbb{R}_{\geq 0}$, which ensures that the closed-loop trajectories do not leave the domain of validity of the DAR, $\mathcal{D}_x$, given in~\eqref{eq:domain_of_validity}. This concludes the proof.~\hfill~$\blacksquare$
\end{pf}

\begin{remark}
    An outstanding feature of the methodology proposed to establish Theorem~\ref{th:theorem_main} is the incorporation of the vector of nonlinearities of the DAR~\eqref{eq:system_dar} into the gain-scheduled control law~\eqref{eq:control_law} and the trigger function~\eqref{eq:trigger_function}.
    While in \cite{reis2025dynamic}, in addition to dealing the continuous-time case, a control law independent of the vector of nonlinearities is designed.
\end{remark}

In the sequel, similarly to \cite{reis2025dynamic}, we derive another co-design condition to design the following gain-scheduled controller
\begin{align}\label{eq:control_law_simplified}
    u_k = K(\hat{x}_k) \hat{x}_k,
\end{align}
which is a particular case of~\eqref{eq:control_law} without the nonlinearity term, that is, $L(\hat{x}_k)=0$, and the following trigger function:
\begin{align}\label{eq:trigger_function_simplified}
    \Gamma(x, e, \pi)=x^{\top} Q_x x - e^{\top} Q_e e -\xi(x, e, \pi),
\end{align}
where $\xi(x, e, \pi)$ is given as in~\eqref{eq:async_term} taking $L(\hat{x})=0$.
The co-design condition is established in the following Corollary.

\begin{corollary} \label{cor:pi_independent_lmis}
    Consider the discrete-time dynamical system in~\eqref{eq:affine_system} represented in a DAR form as in~\eqref{eq:system_dar}, the control law~\eqref{eq:control_law_simplified}, and the dynamic ETM described by~\eqref{eq:dynamic_etm_pi}, \eqref{eq:dynamic_variable_eta_pi}, \eqref{eq:control_law_simplified}.
    Under the conditions of Theorem~\ref{th:theorem_main}, if \eqref{eq:attraction_region_lmi_theorem1} and the following LMIs hold:
    \begin{equation} \label{eq:lmi_sum_theorem1_corollary}
        \sum_{(\mathbf{i}, \mathbf{j})\in \mathscr{P}(\mathbf{m},\mathbf{n})} \Phi_{\mathbf{i} \mathbf{j}} < 0, \quad \forall \mathbf{m},\mathbf{n} \in \mathbb{B}^{r+},  
    \end{equation}
    with
    \begin{align*} 
        \Phi_{\mathbf{i} \mathbf{j}} &= 
        \begin{bmatrix}
        -X & 0 & \Phi_{1,3} & \Phi_{1,4} & X \\
        \star & - \tilde Q_e & \tilde K_{\mathbf{j}}^{\top} \Omega_{3 \mathbf{i}}^{\top} & \tilde K_{\mathbf{j}}^{\top} A_{3 \mathbf{i}}^{\top} & 0 \\
        \star & \star & \operatorname{He}(\Omega_{2 \mathbf{i}} \tilde{Z}) & \tilde Z^{\top} A_{2 \mathbf{i}}^{\top} & 0 \\
        \star & \star & \star & -X & 0 \\
        \star & \star & \star & \star & -\tilde Q_x
        \end{bmatrix}, \\
        \Phi_{1,3} &= X \Omega_{1 \mathbf{i}}^\top + \tilde K_{\mathbf{j}}^\top \Omega_{3 \mathbf{i}}^\top, \quad
        \Phi_{1,4} = X A_{1 \mathbf{i}}^\top + \tilde K_{\mathbf{j}}^\top A_{3 \mathbf{i}}^\top.
    \end{align*}
    then, the origin of the closed-loop system~\eqref{eq:dar_nl_law} with \mbox{$L(x) = 0$} is asymptotically stable. Moreover, the Lyapunov function that certifies the asymptotic stability of the origin is given by~\eqref{eq:lyapunov_candidate}
    and, for a given $\eta_0$, the state trajectories $(x_k,\eta_k)$ with initial condition $x_0$ taken inside of~\eqref{eq:initial_states_region}
    converge asymptotically to the origin $(x,\eta) = (0,0)$ without leaving the region~$\mathcal{R}$ in~\eqref{eq:level_set}, which, in its turn, satisfies $\mathcal{R} \subset \mathcal{D}_x \times \mathbb{R}_{\geq 0}$.
\end{corollary}
\begin{pf}
    The proof follows the same steps as in Theorem~\ref{th:theorem_main}.~\hfill~$\blacksquare$
\end{pf}

\subsection{Optimization problems}

Sufficient co-design conditions for the event-based stabilization of the closed-loop system~\eqref{eq:dar_nl_law} were established in the previous section. This section presents methods to increase the inter-transmission intervals and enlarge the region of attraction estimation of the origin of the closed-loop system. 

To accomplish the objective of increasing the inter-transmission intervals, we
follow a similar approach as in~\cite{CoutinhoP2022110292,Moreira20202720}. 
According to the dynamic ETC scheme~\eqref{eq:dynamic_etm_pi}--\eqref{eq:trigger_function}, a new transmission occurs when 
\begin{equation}
   \mathcal{G}(\zeta_k) > 1 + \frac{\eta_k}{\theta \mathcal{D}(x_k,\pi_k)} - \mathcal{V}(\zeta_k),
\end{equation}
where 
\begin{align}
    \mathcal{G}(\zeta) &= \frac{e^\top Q_e e + \delta^\top Q_{\delta} \delta}{\mathcal{D}(x,\pi)}, \\
    \mathcal{V}(\zeta) &= \frac{\xi(\zeta)}{\mathcal{D}(x,\pi)}, \\
    \mathcal{D}(x,\pi) &= x^\top Q_x x + \pi^\top Q_{\pi} \pi.
\end{align}
Note that $\mathcal{G}(\zeta_k) = \mathcal{V}(\zeta_k)= 0$ at the transmission instants $k = k_j$. Alternatively, transmissions are not triggered while
\begin{equation}
    \mathcal{G}(\zeta_k) \leq \frac{\lambda_{\max}(Q_e) {\| e_k \|}^2 + \lambda_{\max}(Q_{\delta}) {\| \delta_k \|}^2}{\lambda_{\min}(Q_x) {\| x_k \|}^2 + \lambda_{\min}(Q_{\pi}) {\| \pi_k \|}^2},
\end{equation}
that can be recast in the form
\begin{equation}
    \mathcal{G}(\zeta_k) \leq \Lambda \frac{{\| e_k \|}^2 + {\| \delta_k \|}^2} {{\| x_k \|}^2 + {\| \pi_k \|}^2},
\end{equation}
where 
\begin{equation}
    \Lambda = \frac{\max \{\lambda_{\max}(Q_e),\lambda_{\max}(Q_{\delta}) \}}{\min \{ \lambda_{\min}(Q_x),\lambda_{\min}(Q_{\pi}) \}}.
\end{equation}

As a result, additional transmissions are not triggered unless
\begin{equation}
    \frac{{{\| e_k \|^2 + \| \pi_k \|^2}}} {{{\| x_k \|^2 + \| \delta_k \|^2}}} \leq \frac{1}{{\Lambda}} \left({1  + \frac{\eta_k}{\theta \mathcal{D}(x_k,\pi_k)} - \mathcal{V}(\zeta_k)}\right).
\end{equation}
Then, minimizing $\Lambda$ tends to increase the minimum time for $\mathcal{G}(x,e)$ to evolve from $0$ to $1 + \eta/(\theta D(x,\pi)) - \mathcal{V}(\zeta)$. This objective can be achieved by maximizing $\min \{ \lambda_{\min}(Q_x),\lambda_{\min}(Q_{\pi}) \}$ and minimizing $\max \{\lambda_{\max}(Q_e), \allowbreak \lambda_{\max}(Q_{\delta}) \}$.

To achieve the second objective, recall that the set of admissible initial states is given by $\mathcal{R}_0$ in~\eqref{eq:initial_states_region}. 
If the conditions in Theorem~\ref{th:theorem_main} are fulfilled, for $x_0 \in \mathcal{R}_0$, we ensure that the trajectories of the closed-loop system~\eqref{eq:dar_nl_law} asymptotically converge to the origin without leaving the region~$\mathcal{R} \subset \mathcal{D}_x \times \mathbb{R}_{\geq 0}$ in~\eqref{eq:level_set}. This ensures that the closed-loop system operates inside the region of validity of the DAR model~\eqref{eq:system_dar}. Thus, the region of attraction estimation~$\mathcal{R}$ can be maximized by enlarging the region $\mathcal{R}_0$, which can be viewed as a projection of $\mathcal{R}$ onto $\mathcal{D}_x \subset \mathbb{R}^n$.

To simultaneously deal with these two objectives, we consider the $\epsilon$-constraint optimization algorithm, resulting in the objective problem (OP) stated below, that is solved for a given value of $\varepsilon > 0$.
    \begin{align}
        \min_{\tilde Q_x, \tilde Q_e, \tilde Q_{\pi}, \tilde Q_{\delta}, X, \tilde K_{\mathbf{j}}, \tilde L_{\mathbf{j}}} &\sigma \label{eq:optimization_problem}\\
        \textnormal{subject to} \quad &\operatorname{tr}(\tilde Q_x + \tilde Q_e) + \operatorname{tr}(\tilde Q_{\pi} + \tilde Q_{\delta}) < \epsilon, \\
        &\begin{bmatrix} \label{eq:constraint_enlarge_region}
            \sigma I & I \\
            \ast & X
        \end{bmatrix} \geq 0, \\
        &\text{LMIs in}~  \eqref{eq:attraction_region_lmi_theorem1},\eqref{eq:lmi_sum_theorem1}. \nonumber
    \end{align}
    Applying Schur complement and change the variable \mbox{$P=X^{-1}$} to the constraint~\eqref{eq:constraint_enlarge_region}, it follows that \mbox{$P < \sigma I$}.
    Thus, minimizing $\sigma$ increase the region $\mathcal{R}_0$ since \linebreak $V(x) \leq \sigma x^\top x$.

    We state below the optimization problem as a consequence of Corollary~\ref{cor:pi_independent_lmis} corresponding to the case where the control law~\eqref{eq:control_law_simplified} and the trigger function~\eqref{eq:trigger_function_simplified} without the nonlinearity dependence are employed.
    \begin{equation}
        \begin{aligned}
            \min_{\tilde Q_x, \tilde Q_e, X, \tilde K_{\mathbf{j}}} &\sigma \label{eq:optimization_problem_2}\\
        \textnormal{subject to} \quad &\operatorname{tr}(\tilde Q_x + \tilde Q_e) < \epsilon, \\
        &\text{LMIs in}~\eqref{eq:attraction_region_lmi_theorem1}, \eqref{eq:lmi_sum_theorem1_corollary},
        \eqref{eq:constraint_enlarge_region}.
        \end{aligned}
    \end{equation}

By solving the OP~1 in~\eqref{eq:optimization_problem} and the OP~2 in~\eqref{eq:optimization_problem_2} for different values of $\epsilon$, we can obtain different values of $\sigma$. Thus, the pairs $(\epsilon,\sigma)$ obtained from the solutions of those OPs constitute an estimate of the Pareto front.

\section{Numerical results} \label{sec:numerial_results}

In this section, an example is presented to show the effectiveness of the proposed dynamic ETC scheme of discrete-time nonlinear systems in the DAR form. 
Consider the following discretized nonlinear system:
\begin{equation} \label{eq:system_example1}
\begin{aligned}
    x_{(1)k+1} =& x_{(1)k} + T x_{(2)k} \\
    x_{(2)k+1} =& x_{(2)k} + T x_{(1)k} + T x_{(1)k}^3 \\
    &+ 2 T x_{(2)k} + 8 T x_{(2)k}^3 + T u_k,
\end{aligned}
\end{equation}
where $T = 0.1$~s is the sampling period. 

The system~\eqref{eq:system_example1} can be described in the form of a DAR as~\eqref{eq:system_dar} with
\begin{equation} \label{eq:dar_example1}
    \begin{array}{l}
    A_1{=}\begin{bmatrix}1 & T \\ T & 1+2T\end{bmatrix},
    A_2(x){=}\begin{bmatrix}0 & 0 \\ T x_{(1)} & 8T x_{(2)}\end{bmatrix},  A_3{=}\begin{bmatrix}0 \\ T \end{bmatrix}, \\
    \Omega_1(x){=}\begin{bmatrix}x_{(1)} & 0 \\ 0 & x_{(2)}\end{bmatrix},
    \Omega_2{=}\begin{bmatrix}-1 & 0 \\ 0 & -1\end{bmatrix},  
    \Omega_3{=}\begin{bmatrix}0 \\ 0\end{bmatrix} ,\end{array}
\end{equation}
and the vector of nonlinearities 
$\pi = \begin{bmatrix}
    x_{(1)}^2 & x_{(2)}^2
\end{bmatrix}^\top$. It is assumed that the system trajectories are constrained by \mbox{$\mathcal{D}_x = \left\{x \in \mathbb{R}^2:\left|x_{(i)}\right| \leq \overline{x}, i  = 1,2\right\}$}.

The solutions in Fig.~\ref{fig:epsilon_sigma_pareto}, represented by the blue dashed line, were obtained by solving the optimization problem in OP~1 in~\eqref{eq:optimization_problem} with $\overline{x} = 2$ for a set of $20$ values of $\epsilon$ spaced on a logarithmic scale between $10^{-1}$ and $10^5$. For the considered region $\mathcal{D}_x$ with $\overline{x} = 2$, no feasible solutions were obtained from the optimization problem in OP~2 in~\eqref{eq:optimization_problem_2}.
The results show a trade-off between the two objectives, because a larger region of attraction requires a smaller $\sigma$, which forces the system to trigger more events. Then, by exploiting the solutions in the estimated Pareto front, a designer can choose to expand the estimated region of attraction or reduce the number of transmissions. 
\begin{figure} [ht!]
    \centering
    \includegraphics[width=1.0\linewidth]{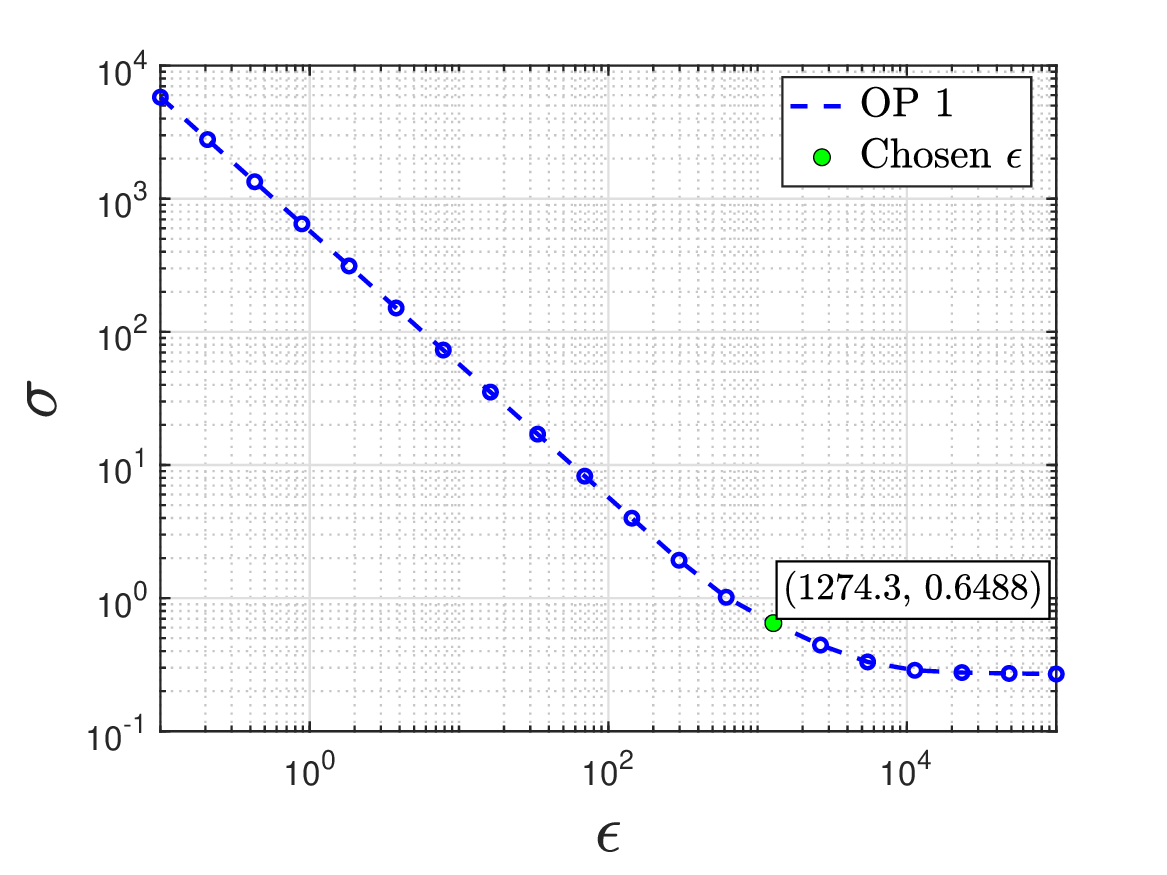}
    \caption{The Pareto optimal solutions found via OP~1 in~\eqref{eq:optimization_problem} in blue. The green marker indicates the point on the Pareto front corresponding to the chosen value of $\epsilon$.}
    \label{fig:epsilon_sigma_pareto}
\end{figure}

To evaluate the conservativeness of OP~2 in~\eqref{eq:optimization_problem_2}, we searched for the maximum value of $\overline{x}$ for which OP~2 was feasible.
The maximum value of $\overline{x}$ region of validity found in which OP~2 in~\eqref{eq:optimization_problem_2} is feasible was $\overline{x}= 1.53$, leading to $\tilde{\mathcal{D}}_x = \left\{x \in \mathbb{R}^2:\left|x_{(i)}\right| \leq 1.53, i  = 1,2\right\}$.
In Fig.~\ref{fig:region_of_sttraction_cor2_cor3}, we show the largest estimated regions of attraction obtained considering $\tilde{\mathcal{D}}_x$ to OP~1 in~\eqref{eq:optimization_problem} and OP~2 in~\eqref{eq:optimization_problem_2} solved with $\epsilon = 10^3$.
This clearly indicates that the inclusion of the nonlinear term $\pi$ in the control law reduces the conservativeness in the proposed approach, as a larger region of attraction estimate was obtained with $\overline{x} = 1.53$, and only OP~1 resulted in feasible solutions with $\overline{x} = 2$.
\begin{figure}[!ht]
    \centering
    \includegraphics[width=8cm]{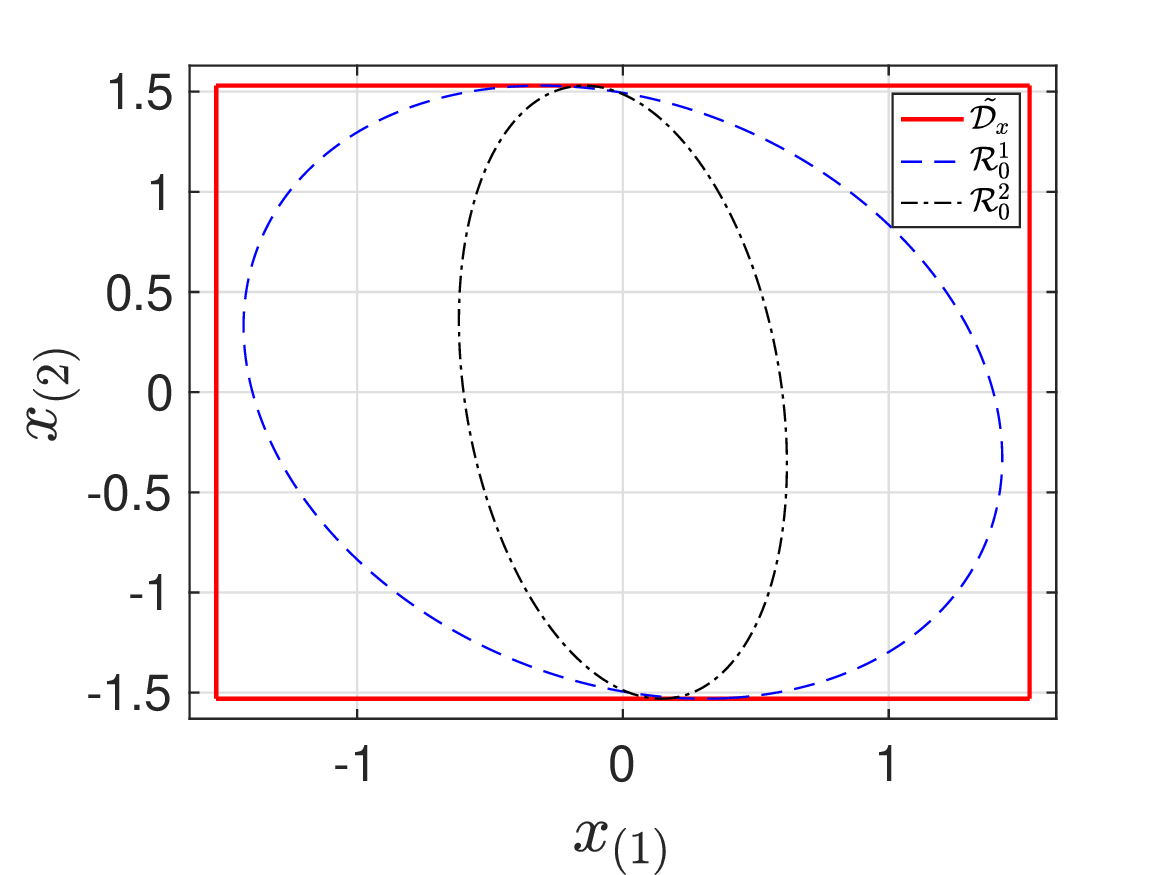}
    \caption{Estimated regions of attraction. Validity region $\tilde{\mathcal{D}}_x$ is depicted as a red solid square. The blue dashed ellipse $\mathcal{R}_0^1$ corresponds to the estimated region for OP~1 in~\eqref{eq:optimization_problem} while the black dash-dotted ellipse $\mathcal{R}_0^2$ corresponds to the estimated region for OP~2 in~\eqref{eq:optimization_problem_2}.}
    \label{fig:region_of_sttraction_cor2_cor3}
\end{figure} 

The estimation of the region of attraction obtained with $\epsilon = 1274.3$, considering the region $\mathcal{D}_x$ is shown in Fig.~\ref{fig:region_of_sttraction_estimation} together with several closed-loop trajectories with initial conditions $x_0$ taken at the border of the region of admissible initial conditions~$\mathcal{R}_0$ and $\eta_0 = 0$. Notice that all trajectories asymptotically converge to the origin $x = 0$.
\begin{figure}[!ht]
    \centering
    \includegraphics[width=8cm]{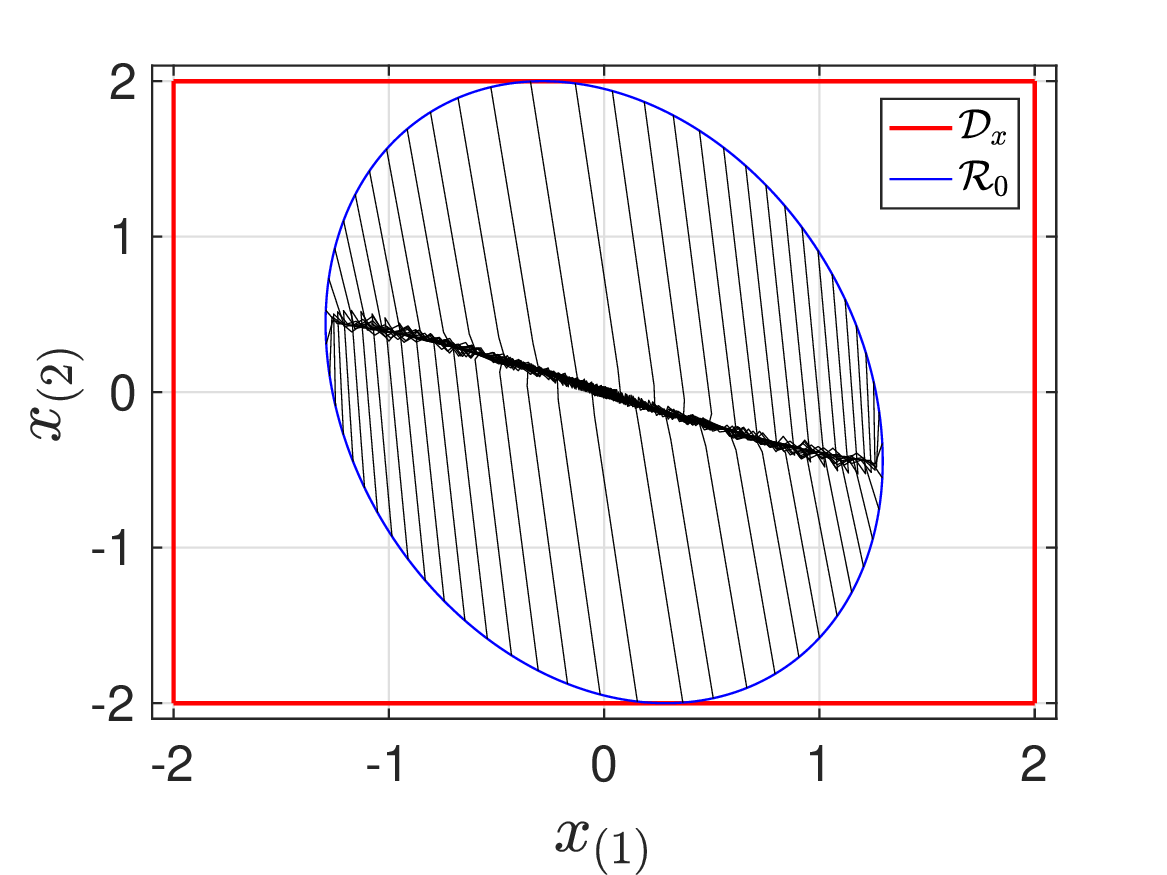}
    \caption{The validity region $\mathcal{D}_x$ and the estimated region of attraction $\mathcal{R}_0$ for OP~1 in~\eqref{eq:optimization_problem}.}
    \label{fig:region_of_sttraction_estimation}
\end{figure} 

We also compare the static and dynamic ETC schemes using the same selected $\epsilon = 1274.3$ and $\overline{x} = 2$, for which only OP~1 is feasible. Table~\ref{tab:exemple02_number_of_events} summarizes the average number of events for both the static and dynamic ETC schemes designed via OP~1 in~\eqref{eq:optimization_problem}. The dynamic ETC scheme was evaluated with different values of $\lambda$ and $\theta$. We simulated the closed-loop system with 60 different initial conditions within the estimated region of attraction. It can be observed that the proposed dynamic ETC scheme effectively reduces the number of events. 
\begin{table}[!ht]
\centering
\caption{Average number of events 
considering
$10$ seconds interval. 
}
\footnotesize
\begin{tabular}{lclll}
\hline Static ETC & 34.79 & & & \\
\hline Dynamic ETC & $\lambda=10^{-3}$ & $\lambda=10^{-2}$ & $\lambda=10^{-1}$\\
\hline 
$\theta=2$ & 34.31 & 34.38 & 34.51\\
$\theta=10$ & 32.25 & 32.31 & 34.67 \\
$\theta=10^2$ & $\mathbf{32.10}$ & 32.41 & 34.78\\
$\theta=10^3$ & 32.75 & 33.03 & 34.63\\
\hline  
\end{tabular}
\label{tab:exemple02_number_of_events}
\end{table}

The time-series of a state trajectory of the closed-loop system, the sampled control input signal, and the inter-event transmission times are depicted in Fig.~\ref{fig:result_states_events_pi_exemple02} for an initial condition set $x_0 = (0.6573,1.4554)$, and parameter $\lambda = 10^{-3}$ and $\theta = 100$. In this simulation, 30 events were transmitted with the dynamic ETC scheme. Using the static version, 35 events are transmitted, and a standard periodic time-triggered scheme transmits 101 events. This illustrates that the proposed dynamic ETC scheme can save communication resources compared to its static counterpart and a standard periodic time-triggered scheduler.
\begin{figure}[!ht]
    \centering
    \includegraphics[width=\columnwidth]{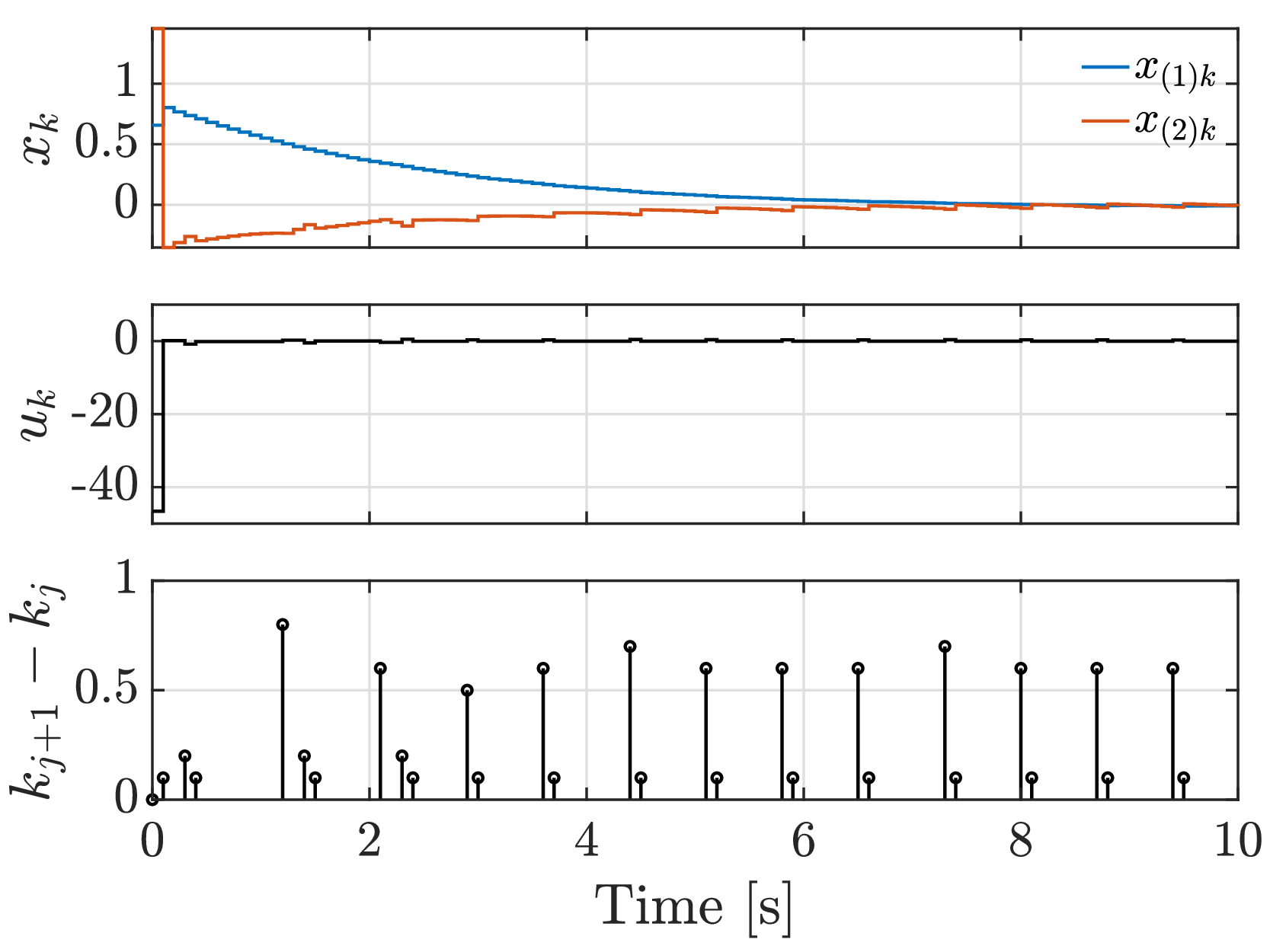}
    \caption{Trajectory of states in a closed-loop with dynamic ETM.
    }
    \label{fig:result_states_events_pi_exemple02}
\end{figure} 

\section{Conclusion} \label{sec:conclusion}

In this paper, a dynamic ETC scheme for discrete-time nonlinear systems based on DAR was addressed. The proposed co-design condition ensured the asymptotic stability of the origin of the closed-loop system. 
The co-design of the ETC scheme with a nonlinear gain-scheduled controller was established through a Lyapunov-based analysis. 
Furthermore, an optimization problem was employed to reduce the number of transmissions and enlarge the region of attraction estimation. 
The numerical example validated the proposed methodology, revealing that the nonlinear gain-scheduled controller provided a less conservative condition than both a standard gain-scheduled controller and a linear one.
Finally, as expected, the dynamic ETC scheme outperformed the static approach by reducing the occurrence of triggering events.

\bibliography{ifacconf}             
                                                   







\end{document}